\documentclass{article}

\usepackage{spconf,amsmath,graphicx}
\usepackage{IEEEtrantools}
\usepackage{amssymb,mathrsfs} 

\usepackage{balance}
\usepackage{color}
\usepackage{amssymb}  
\usepackage{pifont}
\usepackage{array} 
\usepackage{tabularx} 
\usepackage{epstopdf}
\usepackage{dsfont}
\usepackage{bbm, dsfont}
\usepackage{amsthm}

\long\def\comment#1{}

\newfont{\bbb}{msbm10 scaled 700}

\newfont{\bb}{msbm10 scaled 1100}

\newcommand{\RR}{\mbox{\bb R}}
\newcommand{\PP}{\mbox{\bb P}}

\newcommand{\EE}{\mbox{\bb E}}

\newcommand{\yv}{{\bf y}}

\newcommand{\zerov}{{\bf 0}}

\newcommand{\Bm}{{\bf B}}

\newcommand{\Hm}{{\bf H}}
\newcommand{\Id}{{\bf I}}

\newcommand{\Sm}{{\bf S}}

\newcommand{\Wm}{{\bf W}}

\newcommand{\Zm}{{\bf Z}}

\newcommand{\Tt}{\text{T}}

\newcommand{\Hc}{{\cal H}}

\newcommand{\Nc}{{\cal N}}

\newcommand{\RNum}[1]{\uppercase\expandafter{\romannumeral #1\relax}}

\newcommand{\Lambdam}{\hbox{\boldmath$\Lambda$}}

\newcommand{\Sigmam}{\hbox{\boldmath$\Sigma$}}

\newcommand{\trace}{{\hbox{tr}}}

\newcommand{\var}{{\hbox{var}}}

\newcommand{\eqdef}{\stackrel{\Delta}{=}}

\newcommand{\squeezeequ}{\medmuskip=2mu \thinmuskip=1mu \thickmuskip=3mu}

\def\LRT#1#2{\!
\raisebox{.2ex}{$
{{\scriptstyle\;#1}\atop{\displaystyle\gtrless}}
\atop
{\raisebox{-1.25ex}{$\scriptstyle\;#2$}}
$}
\!}

\def\endthebibliography{%
  \def\@noitemerr{\@latex@warning{Empty `thebibliography' environment}}%
  \endlist
}

\newtheorem{theorem}{Theorem}
\newtheorem{lemma}{Lemma}

\title{Learning Requirements for Stealth Attacks}
\name{Ke Sun$^*$, I\~naki Esnaola$^{*\dag}$, Antonia M. Tulino$^{\S\ddag}$, H. Vincent Poor$^{\dag}$
}
\address{$^*$ Dept. of Automatic  Control and Systems Engineering, University of Sheffield, Sheffield S1 3JD, UK\\
$^\S$ Nokia Bell Labs, Holmdel, NJ 07733, USA \\
 $^{\dag}$Dept. of Electrical Engineering, Princeton University, Princeton, NJ 08544, USA  \\
 $^{\ddag}$University degli Studi di Napoli Federico II, 80138 Naples, Italy}

\begin{document}
%
\maketitle
\begin{abstract}
The learning data requirements are analyzed for the construction of stealth attacks in state estimation.
In particular, the training data set is used to compute a sample covariance matrix that results in a random matrix with a Wishart distribution.
The ergodic attack performance is defined as the average attack performance obtained by taking the expectation with respect to the distribution of the training data set.
The impact of the training data size on the ergodic attack performance is characterized by proposing an upper bound for the performance.
Simulations on the IEEE 30-Bus test system show that the proposed bound is tight in practical settings.
\end{abstract}
\begin{keywords}
stealth attacks, data injection attacks, random matrix theory, information theory
\end{keywords}

\section{Introduction}

Data injection attacks \cite{liu_false_2009} are one of the main threats that the smart grid faces.
Attack constructions that exploit the sparsity of the data injection vector have been proposed \cite{kim_strategic_2011} as practical constructions that can disrupt the state estimation performed by the operator.
Distributed attack construction and detection strategies are studied in \cite{cui_coordinated_2012, ozay_sparse_2013, khan_secure_2013,tajer_distributed_2011} where it is shown that the bad data detection procedures put in place by the operator can be defeated by several attackers that control a subset of the sensing infrastructure in the grid.
Modelling the state variables as a random process, attack constructions that exploit the statistical knowledge of the state variables are proposed in \cite{kosut_malicious_2011, esnaola_maximum_2016}.
The addition of probabilistic structure to the state variables opens the door to the definition of information theoretic attacks for which the damage and probability of detection are characterized in terms of information measures \cite{sun_stealth_2018}.
In \cite{Sun_information-theoretic_2017} the assumption of perfect knowledge of the statistics of the state variables is relaxed by considering a training data set to learn the statistics.
Therein, it is numerically shown that the performance of the attack when imperfect knowledge of the statistics is available changes significantly with respect to the case with perfect knowledge.
In this paper, we analytically characterize the impact of the training data size and the correlation between state variables over the attack performance.

\vspace{-0.5em}
\section{System Model}
\vspace{-0.5em}
\subsection{State Estimation and Bad Data Detection}

The measurement model for state estimation with linearized dynamics is given by

\begin{IEEEeqnarray}{c}\label{Equ:DCSE}
Y^{M} = \Hm X^{N} + Z^{M},
\end{IEEEeqnarray}
where $Y^{M} \in \RR^{M }$ is a vector of random variables describing the measurements; $X^{N} \in \RR^{N}$ is a vector of random variables describing the state variables;
$\Hm \in \RR^{M\times N}$ is the linearized Jacobian measurement matrix which is determined by the power network topology and the admittances of the branches;
and $Z^{M} \in \RR^{M} $ is the additive white Gaussian noise (AWGN) with distribution $\Nc(\zerov,\sigma^{2} \Id_{M})$ where $\sigma^2$ is the variance of the error introduced by the sensors \cite{abur_power_2004}, \cite[Chapter 15]{grainger_power_1994}.
The vector of the state variables is assumed to follow a multivariate Gaussian distribution given by
$
  X^{N} \sim \Nc(\zerov,\Sigmam_{X\!X}),
$
where $\Sigmam_{X\!X}$ is the positive-definite covariance matrix of the distribution of the state variables.
The Gaussian assumption for the vector of the state variables is also adopted by \cite{kosut_malicious_2011} and \cite{esnaola_maximum_2016}.
As a result of the linear measurement model in (\ref{Equ:DCSE}), the vector of measurements also follows a multivariate Gaussian distribution denoted by
$
  Y^{M} \sim \Nc(\zerov,\Sigmam_{Y\!Y}),
$
where $\Sigmam_{Y\!Y} = \Hm\Sigmam_{X\!X}\Hm^{\Tt} + \sigma^{2}\Id_{M} $. 

Data injection attacks corrupt the measurements available to the operator by adding an attack vector to the measurements.
The resulting vector of compromised measurements is given by
\vspace{-0.8em}
\begin{IEEEeqnarray}{c}
\label{eq:measurement_model}
Y^{M}_{A} = \Hm X^{N} + Z^{M} + A^{M},
\end{IEEEeqnarray}
where $A^{M} \in \RR^{M}$ is the attack vector and $Y^{M}_{A} \in \RR^{M } $ is the vector containing the compromised measurements \cite{liu_false_2009}.
Following the approach in \cite{sun_stealth_2018} we adopt a Gaussian framework for the construction of the attack vector, i.e.
$
  A^{M} \sim  \Nc (\zerov,\Sigmam_{A\!A}),
$
where $\Sigmam_{A\!A}$ is the covariance matrix of the attack distribution.
The rationale for choosing a Gaussian distribution for the attack vector follows from the fact that for the attack model in (\ref{eq:measurement_model}) the additive attack distribution that minimizes the mutual information between the vector of state variables and the compromised measurements is Gaussian \cite{SA_TIT_13}.
Because of the Gaussianity of the attack distribution, the vector of compromised measurements is distributed as
$
  Y_{A}^{M} \sim \Nc(\zerov,\Sigmam_{Y_{A}\!Y_{A}}),
$
where $\Sigmam_{Y_{A}\!Y_{A}} = \Hm\Sigmam_{X\!X}\Hm^{\Tt} + \sigma^{2}\Id_{M} + \Sigmam_{A\!A} $.


The operator of the power system makes use of the acquired measurements to detect the attack.
The detection problem is cast as a hypothesis testing problem with hypotheses
\begin{IEEEeqnarray}{cl}
\Hc_{0}:  \ & Y^{M} \sim \Nc(\zerov,\Sigmam_{Y\!Y}), \quad \text{versus}  \\
\Hc_{1}:  \ & Y^{M} \sim \Nc(\zerov,\Sigmam_{Y_{A}\!Y_{A}}).
\end{IEEEeqnarray}
The null hypothesis $\Hc_{0}$ describes the case in which the power system is not compromised, while the alternative hypothesis $\Hc_{1}$ describes the case in which the power system is under attack.
The Neyman-Pearson lemma \cite{poor_introduction_1994} states that for a fixed probability of Type \RNum{1} error,
the likelihood ratio test (LRT) achieves the minimum Type \RNum{2} error when compared with any other test with an equal or smaller Type \RNum{1} error.
Consequently, the LRT is chosen to decide between $\Hc_{0}$ and $\Hc_{1}$ based on the available measurements.
The LRT between $\mathcal{H}_{0}$ and $\mathcal{H}_{1}$ takes following form:
\vspace{-0.6em}
\begin{equation}\label{LHRT}
\vspace{-0.25em}
L(\yv) \eqdef \frac{f_{Y^{M}_{A}}(\yv)}{f_{Y^{M}}(\yv)} \ \LRT{\Hc_{1}}{\Hc_{0}} \ \tau,
\vspace{-0.3em}
\end{equation}
where $\yv \in \RR^{M}$ is a realization of the vector of random variables modelling the measurements, $f_{Y_A^M}$ and  $f_{Y^M}$ denote the probability density functions (p.d.f.'s) of $Y_A^M$ and  $Y^M$, respectively, and $\tau$ is the decision threshold set by the operator to meet the false alarm constraint.

\subsection{Information-Theoretic Attacks}


The purpose of the attacker is to disrupt the normal state estimation procedure by
minimizing the information that the operator acquires about the state variables, while guaranteeing that the probability of attack detection is small enough, and therefore, remain concealed in the system. To that end, the attacker aims to minimize the mutual information between the vector of state variables and the vector of compromised measurements denoted by $I(X^{N};Y_{A}^{M})$.
On the other hand, we assess the performance of attack detection by the LRT via the Chernoff-Stein lemma \cite{cover_elements_2012}, which characterizes
the asymptotic exponent of the probability of detection when the number of observations of measurement vectors grows to infinity.
%
In our setting, the Chernoff-Stein lemma states that for any LRT and $\epsilon  \in (0,1/2)$, it holds that
\begin{align}
\label{eq:Chernoff-Stein}
\lim_{n \to \infty} \frac{1}{n} \log \beta_{n}^{\epsilon} = -D(P_{Y^{M}_{A}}\|P_{Y^{M}}),
\end{align}
where $D(\cdot \|\cdot)$ is the Kullback-Leibler (KL) divergence, $\beta_{n}^{\epsilon}$ is the minimum Type II error such that the Type I error $\alpha$ satisfies $\alpha < \epsilon$, and $n$ is the number of $M$-dimensional measurement vectors that are available for the LRT.
Therefore, for the attacker, minimizing the asymptotic detection probability is equivalent to minimizing $D(P_{Y^{M}_{A}}\|P_{Y^{M}})$, where $P_{Y_A^M}$ and  $P_{Y^M}$ denote the probability distributions of $Y_A^M$ and  $Y^M$, respectively.

A stealthy attack construction that combines these two information measures in one cost function is proposed in \cite{Sun_information-theoretic_2017}. Interestingly, the resulting cost function boils down to the effective secrecy proposed in \cite{hou_effective_2014} which can be written as
\begin{IEEEeqnarray}{c}\label{Equ:Steallth_Obj}
 I(X^{N};Y^{M}_{A}) \hspace{-0.1em} +  \hspace{-0.2em} D( P_{{Y}^{M}_{A}}\|P_{Y^{M}}) \hspace{-0.2em} =  \hspace{-0.2em} D( P_{X^{N}Y_{A}^{M}}\|P_{X^{N}}P_{Y^{M}}),\hspace{-0.2em} \IEEEeqnarraynumspace
\end{IEEEeqnarray}
where $P_{X^{N}Y_{A}^{M}}$ is the joint distribution of $X^{N}$ and $Y_{A}^{M}$.
The resulting attack construction problem is equivalent to solving the following optimization problem:
\begin{align} \label{Stealth_Obj}
\underset{A^{M}}{\text{min}} \ D( P_{X^{N}Y_{A}^{M}}\|P_{X^{N}}P_{Y^{M}}).
\end{align}
Under the attack Gaussianity assumption the cost function in (\ref{Equ:Steallth_Obj}) is a function of the attack covariance matrix $\Sigmam_{A\!A}$. Let us define the cost function for the Gaussian case as
\begin{IEEEeqnarray}{l}
\label{eq:stealth_cost_gauss}
f \! \left(\Sigmam_{A\!A}\right) \! \eqdef \! \frac{1}{2} \! \left[\trace(\Sigmam_{Y\!Y}^{-\!1}\Sigmam_{A\!A})\!-\! \log \! |\Sigmam_{A\!A}\!+\!\sigma^{2}\Id_{M}|\!-\!\log \!|\Sigmam_{Y\!Y}^{-\!1}|\right]. \IEEEeqnarraynumspace \squeezeequ
\end{IEEEeqnarray}
It is shown in \cite{sun_stealth_2018} that (\ref{Stealth_Obj}) is a convex optimization problem and that the covariance matrix for the optimal Gaussian attack is $\Sigmam_{A\!A}^{\star} \eqdef \Hm\Sigmam_{X\!X}\Hm^{{\text{\rm T}}}$.

\section{Learning Attack Construction}
The stealth attack construction proposed above requires perfect knowledge of the covariance matrix of the state variables and the linearized Jacobian measurement matrix. In the following we study the performance of the attack when the second order statistics are not perfectly known by the attacker but the linearized Jacobian measurement matrix is known.
We model the partial knowledge by assuming that the attacker has access to a sample covariance matrix of the state variables.
Specifically, the training data consisting of $K$ state variable realizations $\{X^N_{i}\}^{K}_{i=1}$ is available to the attacker. That being the case the attacker computes the unbiased estimate of the covariance matrix of the state variables given by
\vspace{-0.5em}
\begin{equation}\label{Equ:SC}
\Sm_{X\!X} = \frac{1}{K-1} \sum_{i=1}^{K} X^N_{i} (X^N_{i})^{\Tt}.
\vspace{-0.5em}
\end{equation}
The stealth attack constructed using the sample covariance matrix follows a multivariate Gaussian distribution given by
\begin{equation}
\tilde{A}^{M} \sim \Nc (\zerov, \Sigmam_{\tilde{A}\!\tilde{A}}),
\end{equation}
where $\Sigmam_{\tilde{A}\!\tilde{A}} = \Hm\Sm_{X\!X}\Hm^{\Tt}$.

Since the sample covariance matrix in (\ref{Equ:SC}) is a random matrix with central Wishart distribution given by
\begin{align}\label{Equ:Wishart_S_XX}
 \vspace{-0.25em}
\Sm_{X\!X} \sim \frac{1}{K-1}W_{N} (K-1, \Sigmam_{X\!X}),
 \vspace{-0.25em}
\end{align}
the ergodic counterpart of the cost function in (\ref{Equ:Steallth_Obj}) is defined in terms of the conditional KL divergence given by
\begin{equation}
\label{eq:ergodic_cost}
\EE_{\Sm_{X\!X}}\!\left [D\left( P_{X^{N}Y_{A}^{M}|\Sm_{X\!X}}\|P_{X^{N}}P_{Y^{M}}\right)\right].
\end{equation}
The ergodic cost function characterizes the expected performance of the attack averaged over the realizations of training data. Note that the performance using the sample covariance matrix is suboptimal \cite{Sun_information-theoretic_2017} and that the ergodic performance converges asymptotically to the optimal attack construction when the size of the training data set increases.

\section{Upper Bound for Ergodic Attack Performance}

In this section, we analytically characterize the ergodic attack performance defined in (\ref{eq:ergodic_cost}) by providing an upper bound using random matrix theory tools.
Before introducing the upper bound, some auxiliary results on the expected value of the extreme eigenvalues of Wishart random matrices are presented below.
 \vspace{-0.5em}
\subsection{Auxiliary Results in Random Matrix Theory}

\begin{lemma}\label{Pro:VarianceMaxSingular}
Let $\Zm_{L}$ be an $(K-1)\times L$ matrix whose entries are independent standard normal random variables, then
\begin{equation}
\textnormal{\var}\left(s_{max}(\Zm_{L})\right) \leq 1,
\end{equation}
where $\textnormal{\var}\left(\cdot\right)$ denotes the variance and $s_{max}(\Zm_{L})$ is the maximum singular value of $\Zm_{L}$.
\end{lemma}
\begin{proof}
Note that $s_{max}(\Zm_{L})$ is a 1-Lipschitz function of matrix $\Zm_{L}$, the maximum singular value of $\Zm_{L}$ is concentrated around the mean \cite[Proposition 5.34]{vershynin_introduction_2012} given by  $\EE[s_{max}(\Zm_{L})]$. Then for $t\geq 0$, it holds that
\begin{align}
\PP \! \left[ \left|s_{max}(\Zm_{L}) \!- \!\EE[s_{max}(\Zm_{L})]\right| > t \right] &\leq 2\exp\{-t^2/2\}\\
&\leq \exp\{1- t^2/2\}.
\end{align}
Therefore $s_{max}(\Zm_{L})$ is a sub-gaussian random variable with variance proxy $\sigma_{p}^2 \leq 1$.
The lemma follows from the fact that  $\textnormal{\var}\left(s_{max}(\Zm_{L})\right) \leq \sigma_{p}^2$.
\end{proof}

\begin{lemma}\label{Pro:MaxMinEigBound}
Let $\Wm_{L}$ denote a central Wishart matrix distributed as $\frac{1}{K-1} W_{L}(K-1,\Id_{L})$, then the non-asymptotic expected value of the extreme eigenvalues of $\Wm_{L}$ is bounded by
\vspace{-0.5em}
\begin{IEEEeqnarray}{l}\label{Equ:MinEigBound}
\left(1-\sqrt{L/(K-1)}\right)^2\leq \EE[\lambda_{min}(\Wm_{L})]
\end{IEEEeqnarray}
and
\vspace{-0.5em}
\begin{IEEEeqnarray}{l}\label{Equ:MaxEigBound}
\EE[\lambda_{max}(\Wm_{L})] \leq \left(1+\sqrt{L/(K-1)}\right)^2 + 1/(K-1),\IEEEeqnarraynumspace
\vspace{-0.5em}
\end{IEEEeqnarray}
where $\lambda_{min}(\Wm_{L})$ and $\lambda_{max}(\Wm_{L})$ denote the minimum eigenvalue and maximum eigenvalue of $\Wm_{L}$, respectively.
\end{lemma}
\begin{proof}
Note that \cite[Theorem 5.32]{vershynin_introduction_2012}
\begin{equation}
\vspace{-0.1em}
\sqrt{K-1} -\sqrt{L}\leq \EE[s_{min}(\Zm_{L})] \label{Equ:Bound_MinV}
\vspace{-0.1em}
\end{equation}
and
\begin{equation}
\sqrt{K-1} + \sqrt{L} \geq \EE[s_{max}(\Zm_{L})] , \label{Equ:Bound_MaxV}
\end{equation}
where $s_{min}(\Zm_{L})$ is the minimum singular value of $\Zm_{L}$.
Given the fact that $\Wm_{L} = \frac{1}{K-1} \Zm_{L}^{\Tt}\Zm_{L}$, then it holds that
\begin{align}
\EE[\lambda_{min}(\Wm_{L})] &= \!\frac{\EE\!\left[{s_{min}(\Zm_{L})}^2\right]}{K-1}  \!\geq\! \frac{\EE\left[s_{min}(\Zm_{L})\right]^2}{K-1} \label{Equ:Exp_MinEig}
\end{align}
and
\vspace{-1em}
\begin{IEEEeqnarray}{l}
\EE[\lambda_{max}(\Wm_{L})] \!= \!\frac{\EE\!\left[{s_{max}(\Zm_{L})}^2\right]}{K-1}  \!\leq \!\frac{\EE\left[s_{max}(\Zm_{L})\right]^2 \hspace{-0.2em} + 1}{K-1}, \IEEEeqnarraynumspace \squeezeequ \label{Equ:Exp_MaxEig}
\end{IEEEeqnarray}
where (\ref{Equ:Exp_MaxEig}) follows from Lemma \ref{Pro:VarianceMaxSingular}.
Combining (\ref{Equ:Bound_MinV}) with (\ref{Equ:Exp_MinEig}), and (\ref{Equ:Bound_MaxV}) with (\ref{Equ:Exp_MaxEig}), respectively, yields the lemma.
\end{proof}
\vspace{-1em}
\subsection{Main Result}
\vspace{-0.25em}
The ergodic attack performance is given by
\vspace{-0.5em}
\begin{IEEEeqnarray}{ll}
\label{eq:exp_stealth_cost_gauss}
\nonumber
&\EE\left[f(\Sigmam_{\tilde{A}\!\tilde{A}} )\right]\\
& \ = \! \frac{1}{2}\EE\!\left[\trace(\Sigmam_{Y\!Y}^{-\!1}\Sigmam_{\tilde{A}\!\tilde{A}})\!-\!\log |\Sigmam_{\tilde{A}\!\tilde{A}}\!+\!\sigma^{2}\Id_{M}|\!-\!\log |\Sigmam_{Y\!Y}^{-\!1}|\right] \nonumber \\
& \ = \! \frac{1}{2}\!\Big(\!\trace \! \left(\Sigmam_{Y\!Y}^{-\!1}\Sigmam^\star_{A\!A}\right)\!-\!\log \! \left|\!\Sigmam_{Y\!Y}^{-\!1}\!\right|\!-\! \EE\!\left[\log \! |\Sigmam_{\tilde{A}\!\tilde{A}}\!+\!\sigma^{2}\Id_{M}|\right]\!\Big).\IEEEeqnarraynumspace
\vspace{-0.25em}
\end{IEEEeqnarray}
The assessment of the ergodic attack performance boils down to evaluating the last term in (\ref{eq:exp_stealth_cost_gauss}). Closed form expressions for this term are provided in \cite{alfano_capacity_2004} for the same case considered in this paper. However, the resulting expressions are involved and are only computable for small dimensional settings. For systems with a large number of dimensions the expressions are computationally prohibitive. To circumvent this challenge we propose a lower bound on the term that yields an upper bound on the ergodic attack performance. Before presenting the main result we provide the following auxiliary convex optimization result.


\begin{lemma} \label{Lemma:logdet_Inv_Wishart}
Let $\Wm_{p}$ denote a central Wishart matrix distributed as $\frac{1}{K-1} W_{p}(K-1,\Id_{p})$ and let $\Bm = \textnormal{diag} (b_{1}, \dots, b_{p})$ denote a positive definite diagonal matrix.
Then
\begin{align}
\vspace{-2em}
\EE\left[\log \left|\Bm + \Wm_{p}^{-1} \right|\right] \geq \sum_{i=1}^{p} \log\left(b_{i} + 1/x_{i}^{\star}\right),
\vspace{-4em}
\end{align}
where $x_{i}^{\star}$ is the solution to the convex optimization problem given by
\vspace{-1em}
\begin{IEEEeqnarray}{ll}
\underset{\left \lbrace x_{i}\right\rbrace_{i=1}^p}{\textnormal{min}} \ & \sum_{i=1}^{p} \log\left( b_{i} + 1/x_{i} \right)\label{Equ:BUPN_1}\\
s.t. \ \ & \sum_{i=1}^{p} x_{i}  = p \label{Equ:BUPN_4} \\
& \textnormal{max}\left(x_{i}\right) \leq  \left(1+\sqrt{p/(K-1)}\right)^2 + 1/(K-1)  \label{Equ:BUPN_5} \IEEEeqnarraynumspace \\
\vspace{-0.25em}
& \textnormal{min}\left( x_{i}\right) \geq  \left(1-\sqrt{p/(K-1)}\right)^2 .\label{Equ:BUPN_6}
\vspace{-0.5em}
\end{IEEEeqnarray}
\end{lemma}

\begin{proof}
Note that
\begin{align}
\EE\left[\log \left|\Bm + \Wm_{p}^{-1} \right|\right] &= \sum_{i=1}^{p} \EE \left[\log \left( b_{i} + \frac{1}{\lambda_{i}(\Wm_{p})}\right)\right] \label{Equ:Convex_Opt_1}\\
&\geq \sum_{i=1}^{p} \log \left( b_{i} +  \frac{1}{\EE [\lambda_{i}(\Wm_{p})]}\right) \label{Equ:Convex_Opt_2}
\end{align}
where in (\ref{Equ:Convex_Opt_1}), $\lambda_{i}(\Wm_{p})$ is the $i$-th eigenvalue of $\Wm_{p}$ in decreasing order;
(\ref{Equ:Convex_Opt_2}) follows from Jensen's inequality due to the convexity of $\log\left(b_{i} +\frac{1}{x}\right)$ for $x > 0$.
Constraint (\ref{Equ:BUPN_4}) follows from the fact that $\EE [\textnormal{trace} (\Wm_{p})] = p$, and constraints  (\ref{Equ:BUPN_5}) and  (\ref{Equ:BUPN_6}) follow from Lemma \ref{Pro:MaxMinEigBound}. This completes the proof.

%
\end{proof}

The following theorem provides a lower bound for the last term in  (\ref{eq:exp_stealth_cost_gauss}), and therefore, it enables us to
characterize the ergodic attack performance.

\begin{theorem}\label{Theorem:NonAsym_1}
Let $\Sigmam_{\tilde{A}\!\tilde{A}}=\Hm\Sm_{X\!X}\Hm^{\Tt}$ with $\Sm_{X\!X}$ distributed as $\frac{1}{K-1}W_{N} (K-1, \Sigmam_{X\!X})$ and denote by $\Lambdam_{p} = \textnormal{diag} (\lambda_{1}, \dots, \lambda_{p})$ the diagonal matrix containing the nonzero eigenvalues in decreasing order.
Then
\begin{align}
&\EE\!\left[\log \! |\Sigmam_{\tilde{A}\!\tilde{A}}\!+\!\sigma^{2}\Id_{M}|\right] \nonumber \\
&\quad \geq  \left(\sum_{i=0}^{p-1} \psi (K-1-i) \right)- p\log(K-1)\nonumber \\
& \qquad \quad + \sum_{i=1}^{p} \log\left(\frac{\lambda_{i}}{\sigma^{2}} + \frac{1}{\lambda_{i}^{\star}}\right)+2M\log\sigma,
\end{align}
where $\psi (\cdot)$ is the Euler digamma function, $p=\textnormal{rank}(\Hm\Sigmam_{X\!X}\Hm^{\Tt})$, and
$\lbrace\lambda_{i}^{\star}\rbrace_{i=1}^p$ is the solution to the optimization problem given by (\ref{Equ:BUPN_1}) - (\ref{Equ:BUPN_6}) with $b_{i} = \frac{\lambda_{i}}{\sigma^{2}}, \textnormal{for}\; i = 1, \dots,p$.
\end{theorem}

\begin{proof}
We proceed by noticing that
\begin{IEEEeqnarray}{ll}
\interdisplaylinepenalty = 0
&\EE\!\left[\log \! |\Sigmam_{\tilde{A}\!\tilde{A}}\!+\!\sigma^{2}\Id_{M}|\right] \nonumber\\
& \ =\EE\left[\log \left|\frac{1}{(K-1)\sigma^{2}}  \Zm_{M}^{\Tt} \Lambdam\Zm_{M} +\Id_{M}\right|\right] + 2M \log \sigma\label{Equ:LBN_1} \IEEEeqnarraynumspace\\
& \ = \EE\left[\log \left|\frac{\Lambdam_{p}}{\sigma^{2}}  \frac{\Zm_{p}^{\Tt} \Zm_{p}}{K-1} +\Id_{M}\right|\right] + 2M \log \sigma\label{Equ:LBN_2}\\
& \ = \! \EE  \! \left[\!\log  \! \left|\frac{\Zm_{p}^{\Tt} \Zm_{p}}{K-1}\right|  \! +  \! \log \!\left|\frac{\Lambdam_{p}}{\sigma^{2}} \! \! + \! \left( \frac{\Zm_{p}^{\Tt} \Zm_{p}}{K-1}\right)^{-\!1}  \! \right|\right] \!+  \!2M\log \sigma \label{Equ:LBN_3}
\IEEEeqnarraynumspace\\
&\ \geq  \left(\sum_{i=0}^{p-1} \psi (K-1-i) \right)- p\log(K-1)\nonumber \\
& \qquad \quad  + \sum_{i=1}^{p} \log\left(\frac{\lambda_{i}}{\sigma^{2}} + \frac{1}{\lambda_{i}^{\star}}\right)+2M\log\sigma, \label{Equ:LBN_5}
\end{IEEEeqnarray}
where in (\ref{Equ:LBN_1}), $\Lambdam$ is a diagonal matrix containing the eigenvalues of $\Hm\Sigmam_{X\!X}\Hm^{\Tt}$ in decreasing order;
(\ref{Equ:LBN_2}) follows from the fact that $p=\textnormal{rank}(\Hm\Sigmam_{X\!X}\Hm^{\Tt})$;
(\ref{Equ:LBN_5}) follows from \cite[Theorem 2.11]{tulino_random_2004} and Lemma \ref{Lemma:logdet_Inv_Wishart}.
This completes the proof.
\end{proof}

\begin{theorem}\label{Theorem:NonAsym_2}
The ergodic attack performance given in (\ref{eq:exp_stealth_cost_gauss}) is upper bounded by
\begin{IEEEeqnarray}{lll}
&\EE\left[f(\Sigmam_{\tilde{A}\!\tilde{A}} )\right] &\leq  \frac{1}{2}\!\Bigg(\!\trace \! \left(\Sigmam_{Y\!Y}^{-\!1}\Sigmam^\star_{A\!A}\right)\!-\!\log \! \left|\!\Sigmam_{Y\!Y}^{-\!1}\!\right| - 2M \log \sigma \Bigg. \IEEEeqnarraynumspace \IEEEnonumber\\
&& \qquad - \bigg(\sum_{i=0}^{p-1} \psi (K-1-i)\! \bigg) \!+\! p\log(K-1) \IEEEnonumber \\
&& \qquad  \quad \Bigg. -\sum_{i=1}^{p} \log\left(\frac{\lambda_{i}}{\sigma^{2}}+\frac{1}{\lambda_{i}^{\star}}\right)\Bigg).
\end{IEEEeqnarray}
\end{theorem}
\begin{proof}
The proof follows immediately from combing Theorem \ref{Theorem:NonAsym_1} with (\ref{eq:exp_stealth_cost_gauss}).
\end{proof}

\section{Numerical Results}
\begin{figure}[t!]
  \centering
  \includegraphics[scale =0.4]{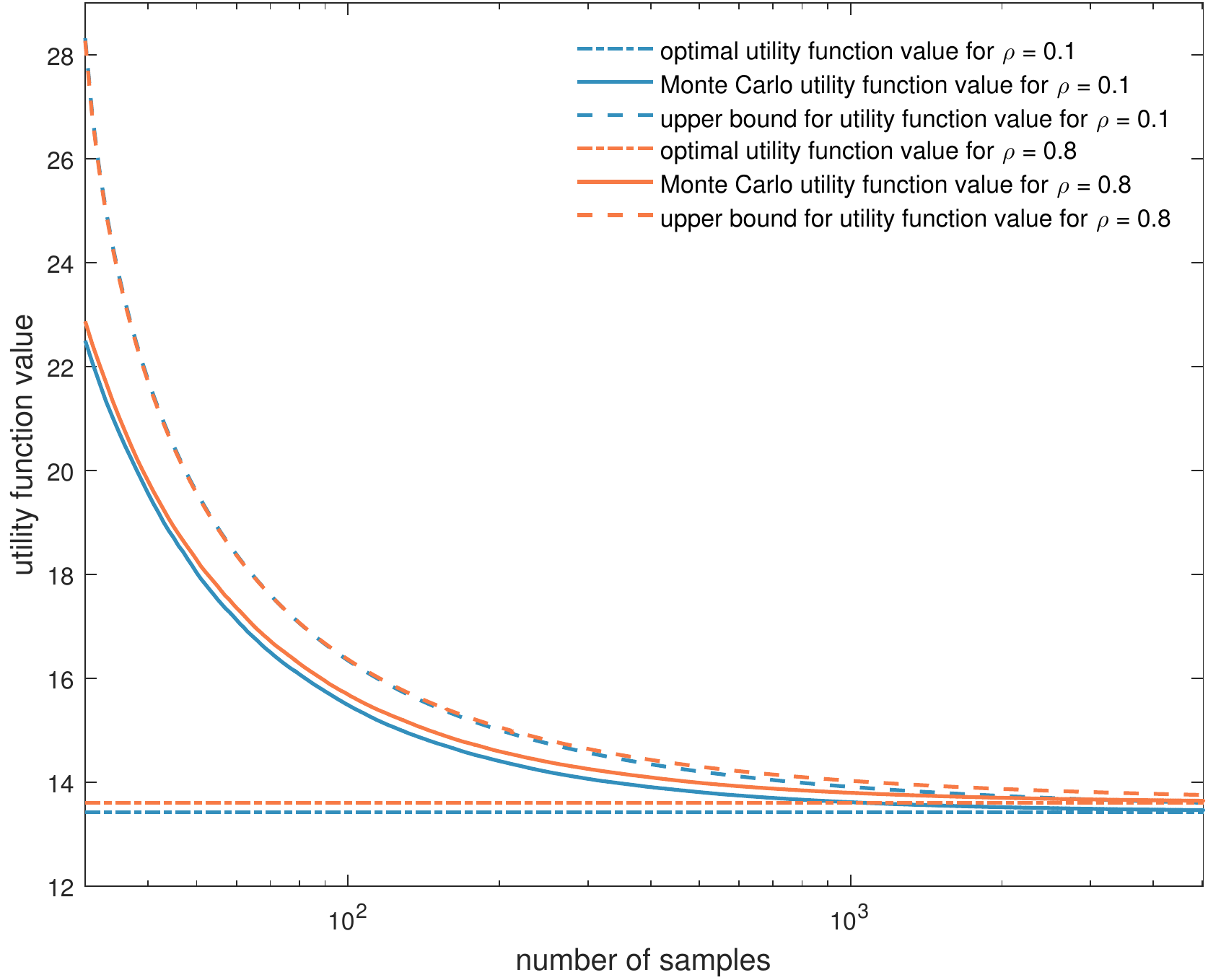}\\
  \caption{Performance of the upper bound in Theorem \ref{Theorem:NonAsym_2} as a function of number of sample for $\rho =0.1$ and $\rho=0.8$ when $\textnormal{SNR} = 20 \ \textnormal{dB}$.}\label{Fig:UB_Rho_30_SNR20_NonAsy}
\end{figure}
The numerical results are obtained on the IEEE 30-Bus test system where
the Jacobian matrix $\Hm$ is obtained using MATPOWER \cite{zimmerman_matpower:_2011}.
For the construction of the stealth attack the covariance matrix of the state variables is chosen to be a Toeplitz matrix with exponential decay parameter $\rho$ as in \cite{esnaola_maximum_2016}.
Specifically, the Toeplitz matrix of dimension $N \times N$ with exponential decay parameter $\rho$ is given by $\Sigmam_{X\!X}=[s_{ij}=\rho^{|i-j|}; i, j =1, 2, \ldots, N].$
We define the Signal-to-Noise Ratio (SNR) as
\begin{equation}
\textnormal{SNR}=10\log_{10}\left(\frac{\trace{(\Hm\Sigmam_{X\!X}\Hm^\textnormal{T}})}{M\sigma^2}\right).
\end{equation}
Fig.\ref{Fig:UB_Rho_30_SNR20_NonAsy} depicts the upper bound in Theorem \ref{Theorem:NonAsym_2} as a function of number of samples for $\rho =0.1$ and $\rho=0.8$ when $\textnormal{SNR} = 20 \ \textnormal{dB}$.
Interestingly, the upper bound in Theorem \ref{Theorem:NonAsym_2} is tight for large values of the training data set size for all values of the exponential decay parameter determining the correlation.

\balance
\bibliographystyle{IEEEbib}
\bibliography{reference}

\end{document}